\documentclass{article}

\usepackage[utf8]{inputenc} %
\usepackage[T1]{fontenc}    %
\usepackage{hyperref}       %
\usepackage{url}            %
\usepackage{booktabs}       %
\usepackage{amsfonts}       %
\usepackage{nicefrac}       %
\usepackage{microtype}      %
\usepackage{natbib}
\usepackage{amsmath}
\usepackage{algorithm}
\usepackage[noend]{algpseudocode}
\usepackage{enumerate}
\usepackage{enumitem}
\usepackage{graphicx}
\usepackage{xcolor}
\usepackage{verbatim}
\usepackage{adjustbox}
\usepackage{multirow}

\usepackage{multicol}

\usepackage{comment}

\usepackage{amsthm}
\usepackage{amsmath}
\usepackage{amssymb}
\usepackage{bm}

\usepackage{cleveref}
\usepackage{hyperref}

\usepackage{thmtools}
\usepackage{thm-restate}

\usepackage[margin=1in]{geometry}

\newtheorem{theorem}{Theorem}[section]

\newtheorem{definition}{Definition}[section]

\newcommand{\prob}[2][]{\mathrm{Pr}_{#1}\left[ #2 \right]}
\newcommand{\expectation}[2][]{\mathop{\mathbb{E}}\displaylimits_{#1}\left[ #2 \right]}
\newcommand{\var}[2][]{\mathop{\mathbb{V}}\displaylimits_{#1}\left[ #2 \right]}

\newcommand{\md}{\mathfrak{D}}
\newcommand{\emd}{\tilde{\md}}

\newcommand{\probop}[1]{\mathop{\mathbb{#1}}\displaylimits}
\newcommand{\Probop}{\probop{P}}
\newcommand{\Expop}{\probop{E}}
\newcommand{\Eexpop}{\hat{\probop{E}}}
\newcommand{\Varop}{\probop{V}}
\newcommand{\Evarop}{\hat{\probop{V}}}

\newcommand{\Probp}[1]{\Probop\left(#1\right)}

\newcommand{\Expwrt}[2]{\Expop_{#1}\left[#2\right]}
\newcommand{\Varwrt}[2]{\Varop_{#1}\left[#2\right]}

\newcommand{\Eexpp}[1]{\Eexpop\left[#1\right]}
\newcommand{\Evarp}[1]{\Evarop\left[#1\right]}

\newcommand{\distributed}{\thicksim}
\newcommand{\rand}[1]{\mathcal{#1}}

\newcommand{\HC}{\mathcal{H}}
\newcommand{\EMD}[3]{\emd_{#1}(#2,#3)}
\newcommand{\RC}[3]{\mathfrak{R}_{#1}(#2,#3)}

\newcommand{\setbuilder}[2]{\left\{ #1 \ | \ #2 \right\}}

\newcommand{\abs}[1]{\left\lvert #1 \right\rvert}
\newcommand{\norm}[1]{\left\lVert #1 \right\rVert}
\newcommand{\R}{\mathbb{R}}

\title{Uniform Convergence Bounds for Codec Selection}

\author{Clayton Sanford, Cyrus Cousins, Eli Upfal}

\newcommand{\clay}[1]{\todo[color=red!40]{Clayton: #1}}

\begin{document}

\maketitle

\begin{abstract}

We frame the problem of selecting an optimal audio encoding scheme as a supervised learning task.  Through uniform convergence theory, we guarantee approximately optimal codec selection while controlling for selection bias.  We present rigorous statistical guarantees for the codec selection problem that hold for arbitrary distributions over audio sequences and for arbitrary quality metrics.  Our techniques can thus balance sound quality and compression ratio, and use audio samples from the distribution to select a codec that performs well on that particular type of data.  The applications of our technique are immense, as it can be used to optimize for quality and bandwidth usage of streaming and other digital media, while significantly outperforming approaches that apply a fixed codec to all data sources.

\end{abstract}
\section{Introduction}

Large amounts of digital media are downloaded and streamed on the web by a growing number of users, often across long distances and over low-bandwidth mobile connections.  Transmitting audio and video signals with perfect fidelity is generally prohibitively expensive, so often lossy encoder-decoders (codecs) are employed, which compress media to a much smaller size, but only approximate the original signal.  Depending on the intended use of the data, different types of signal degradation may or may not be acceptable, and no-free-lunch analysis tells us that no codec will be optimal for all data distributions or all notions of quality.  In this paper, we analyze the problem of \emph{statistically significant codec selection}, where the goal is to, given a set of codecs, a notion of quality, and a sample of media sequences, select a codec that with \emph{high probability}, is \emph{approximately optimal} over the distribution from which the media sequences were drawn.

We allow for arbitrary objective quality metrics, which for example may consider both compression ratios and perceptual quality metrics.  We want to select the optimal codec for a particular data distribution with respect to some objective, thus both the data distribution and the objective influence this choice.  For example, a speech codec could have high quality and low size on speech, but poor quality on music, while a music codec may produce larger encodings, but with better quality on music.  Depending on our choice of metric, we may find that the speech codec is optimal on speech data, wheras the music codec is optimal for music data.

We do not assume knowledge of the distribution from which sequences are drawn, requiring only a \emph{sample} from this distribution.  Under this framing, codec selection becomes a \emph{supervised learning problem}, as we are given data, and tasked with selecting a codec that performs well over the data distribution.  When selecting between a family $\HC$ of codecs, multiple hypothesis testing issues arise. We control these issues through \emph{uniform convergence theory}, through which we guarantee that with high probability, the empirical performance of all codecs over our sample simultaneously approximates their performance over the data distribution.  Specifically, we guarantee that with high probability, the selected codec $\hat{h}$ performs approximately as well on average as any $h \in \HC$ over the distribution.  We also handle more complicated constrained queries, which seek the optimal codec that satisfies some property, for instance the highest-quality codec that achieves a $50\%$ compression ratio or better.

We propose two general meta-algorithms that probabilistically bound the objective function and find this optimal encoding function by process of elimination.  Each may be instantiated with a particular type of uniform convergence bound, which we describe in detail for bounds based on the \emph{empirical maximum discrepancy} \citep{bartlett-emd} (both under a boundedness assumption, and in an unbounded setting, with a novel asymptotic uniform convergence tail bound), as well as several types of \emph{union bound}. Our first meta algorithm is \textit{Global Sampling}, which bounds the means of each criteria and selects a scheme if its objective is optimal with high confidence --- that is, if its confidence interval does not overlap with that of any other scheme. However, this algorithm can be wasteful of computation and data, as it applies all codecs to all samples. We also introduce the \textit{Progressive Sampling with Pruning} (PSP) algorithm, which adaptively prunes provably suboptimal codecs over time, and terminates when a desired approximation threshold is met, thus reducing computation and data consumption. Pruning also improves the statistical performance, as PSP employs uniform guarantees over subsets of the original codec family, and thus obtains the benefits of \emph{localization} \citep{bartlett2005local}. 

The bounds and algorithms presented here encourage customizable and resource-efficient ways to make a statistically confident data-driven decision that is appropriate to the task at hand. They allow users to ensure selection of an optimal encoding scheme given a wide range of preferences and to navigate the statistical trade-offs of encoding. Because storage of data sequences on disk space and personal devices are often more limited than server computation, users ensure that their restrictions on storage space and reconstruction quality are addressed with statistical confidence. Further, our progressive sampling algorithm avoids wasted computation time, and users can tweak the confidence parameter to balance the trade-off between confidence and size of training sample.

We compare the effectiveness of union and uniform convergence bounds at establishing tight bounds and pruning codecs with these algorithms. We examine both traditional finite-sample uniform convergence bounds and novel asymptotic uniform convergence bounds bounds, which bound variances and leverage central limit theorems (classical and Martingale \citep{brown1971martingale}). Using a dataset of segments from audio books with a variety of MP3 encoding schemes, we find that the asymptotic bounds outperform the exact bounds. Because we use few encoding schemes, the union bounds produce tighter bounds experimentally; however, the uniform-convergence bounds will outperform the union bounds with a larger class of codecs. 

In summary, the contributions of our paper are as follows:
\begin{enumerate}[wide, labelwidth=!, labelindent=0pt]
\setlength{\itemsep}{0pt}
\setlength{\parskip}{0pt}
\item We frame \emph{codec selection} as a \emph{supervised learning problem}.
\item We apply \textit{uniform converge theory} to obtain generalization bounds for the codec selection problem, which we contrast with union bounds.
\item We show novel asymptotic uniform convergence bounds for \emph{unbounded} variables under a \emph{finite variance} assumption.
\item We define the \emph{progressive sampling with pruning} algorithm for \emph{adaptive codec search}, with global generalization guarantees.
\item We present experiments showing that our asymptotic bounds outperform the finite-sample bounds.  Our experiments use an insufficient number of codecs for uniform convergence bounds to outperform union bounds, though we argue that this trend reverses for larger codec families.
\end{enumerate}

\section{Background}

Here we present relevant background in statistics and uniform convergence theory, as well as background on audio compression.

\subsection{Uniform Converge and the Empirical Maximum Discrepancy}

A core goal of \emph{uniform convergence theory} is to obtain bounds on the expectation and tails of the \emph{supremum deviation}, defined as
\[
\sup_{h \in \HC} \Expwrt{x' \distributed \mathcal{D}}{h(x')} - \Eexpp{h(\bm{x})} \enspace
\]
given distribution $\rand{D}$ over $\mathcal{X}$, sample $\bm{x} \distributed \rand{D}^m$, and function family $\HC \subseteq \mathcal{X} \to \R$. Related quantities, such as the supremum over \emph{absolute differences} between true and empirical expectation, are also of interest. In the codec selection domain, we aim to estimate the means of certain criteria by computing intervals that contain the true mean with high probability. Bounding on the difference between true and empirical means allow us to be certain that what works well empirically will with high probability work approximately as well on similarly distributed data.

Classical asymptotic results abound and are routinely used in applied statistics. These results, such as the Glivenko-Cantelli theorem \citep{cantelli1933sulla,glivenko1933sulla}, hold asymptotically as the sample size tends to infinity. The conclusions we may draw from finite-sample results, such as the DKW inequality \citep{dvoretzky1956asymptotic,massart1990tight}, which may be viewed as a special-case of the Vapnik-Chervonenkis \citep{chervonekis} inequality hold for finite-samples.  Because asymptotically negligible terms cannot be neglected, asymptotic results are often substantially harder to prove. However, these results are generally far more valuable to scientific and empirical inquiry, as they have the potential to yield $\epsilon$-$\delta$ probabilistic bounds for \emph{finite samples}.  Early results --- such as VC-dimension --- were \emph{distribution-free}, meaning the bounds applied for any $\rand{D}$ over $\mathcal{X}$. Because distribution-free bounds are necessarily worst-case, more recent work has favored \emph{distribution-dependent} and \emph{data-dependent} bounds, which may depend on $\rand{D}$ or properties of $\rand{D}$ estimated through $\bm{x}$.

\citet{bartlett-emd} introduced Empirical Maximum Discrepancy (EMD) to obtain uniform convergence bounds, which are similar to the better-known Rademacher Complexity (RC) \citep{koltchinskii2001rademacher,bartlett-rademacher} bounds. Both of these sample complexity measurements are distribution-dependent.  While they are primarily used for obtaining \emph{generalization bounds} for \emph{supervised learning} in the statistical learning theory literature, they have also been applied to \emph{unsupervised learning} and \emph{sampling} problems \citep{riondato2015mining,riondato2018abra}.

\begin{definition}[Empirical Maximum Discrepancy {[EMD]}]
Given a sample $\bm{x} \in \mathcal{X}^m$, where $2 | m$, and a function family $\HC \subseteq \mathcal{X} \to \R$, the EMD is defined as
\[
\EMD{m}{\HC}{\bm{x}} \doteq \sup_{h \in \mathcal{H}}\frac{1}{m}\sum_{i=1}^m (-1)^i h(\bm{x}_i)
\]
\end{definition}

We use the EMD generalization bounds of \cite{bartlett-emd} to bound the difference between observed and true errors, which we use in an online manner in our algorithm to bound the values of various criteria of our encoding schemes.
\begin{theorem}[Finite-Sample EMD Uniform Convergence Bounds]
\label{thm:finite-sample-emd}
Let $\mathcal{H}$ be a set of functions representing the errors of hypotheses such that $h: \mathcal{X} \rightarrow [0,1]$, $\forall h \in \mathcal{H}$, where $\mathcal{X}$ represents the feature space, and let $\bm{x} \distributed \rand{D}^m$, and $\delta \in (0, 1)$. With probability at least $1 - \delta$,
\[
\sup_{h \in \HC}
\expectation[x\distributed\mathcal{D}]{h(x)} - \frac{1}{m} \sum_{i=1}^{m} h(\bm{x}_i) \leq  2 \emd_m(\mathcal{F}, \bm{x}) + 3\sqrt{\frac{\ln(\frac{1}{\delta})}{2m}} \enspace.
\]
Furthermore, also with probability at least $1 - \delta$, 
\[
\sup_{h \in \HC} \abs{\expectation[x\distributed\mathcal{D}]{h(x)} - \frac{1}{m} \sum_{i=1}^{m} h(\bm{x}_i)} \leq 2 \emd_m(\mathcal{F}, \bm{x}) + 3\sqrt{\frac{\ln(\frac{2}{\delta})}{2m}} \enspace.
\]
\end{theorem}
These results are trivial consequences of a symmetrization inequality and McDiarmid's finite difference inequality \citep{mcdiarmid1989method}.

So far, we have made the \emph{boundedness} assumption, where values must lie on $[0, 1]$.  Through linear transformation, this is sufficient to handle any bounded interval, which through Hoeffding's inequality or McDiarmid's inequality results in finite-sample tail bounds.  However, often we can't assume values or bounded, and even if they are, it may be the case that they are concentrated within a tiny region of their interval, and consequently we obtain extremely wide confidence intervals.  We address both of these issues with our asymptotic uniform convergence bounds, presented in \cref{sec:asymptotic}.

\subsection{On Union Bounds}

We now present analogues to the above uniform convergence bounds, where the supremum deviation is bounded with a \emph{union bound} over applications of Hoeffding's inequality \citep{hoeffding1963probability} (for bounded random variables) or a Gaussian Chernoff bound (under an asymptotic normality assumption).

\begin{theorem}[Finite-Sample Hoeffding Union Bound]
\label{thm:hoeffding-union}

Suppose $\HC \subseteq \mathcal{X} \to [0, 1]$.  The following holds with probability at least $1 - \delta$ over $\bm{x} \distributed \rand{D}^m$:

\[
\sup_{h \in \HC} \abs{\Expwrt{x' \distributed \rand{D}}{h(x')} - \frac{1}{m}\sum_{i=1}^m h(\bm{x}_i)} \leq \sqrt{\frac{\ln\left(\frac{2\abs{\HC}}{\delta}\right)}{2m}}
\]

\end{theorem}

This result is a straightforward consequence of Hoeffding's inequality for each $h \in \HC$, which simultaneously hold due to a union bound.

\begin{theorem}[Asymptotic Gaussian-Chernoff Union Bound]
\label{thm:gaussian-chernoff}

Suppose $\HC \subseteq \mathcal{X} \to \R$, and $\bm{x} \distributed \rand{D}^m$:  Suppose also that $\Varwrt{x' \distributed \rand{D}}{h(x')}$ is finite for all $h \in \HC$.  Now, take $\hat{\sigma}$ to be the maximum empirical variance among all $h \in \HC$ over $\bm{x}$.  The following holds, asymptotically in $m$, with probability at least $1 - \delta$ over choice of $\bm{x}$:

\[
\sup_{h \in \HC} \abs{\Expwrt{x' \distributed \rand{D}}{h(x')} - \frac{1}{m}\sum_{i=1}^m h(\bm{x}_i)} \leq 2\hat{\sigma}\sqrt{\frac{\ln\left(\frac{2\abs{\HC}}{\delta}\right)}{2m}}
\]

\end{theorem}

This result is a trivial consequence of a Gaussian Chernoff bound for each $h \in \HC$, which hold together by union bound.  Naturally, one may wonder whether the uniform convergence bounds are tighter than the union bounds.  We expect the EMD bounds to outperform the union bounds when bounding the generalization of large hypothesis classes because the union bounds scale with increases in $\abs{\mathcal{H}}$.  Furthermore, the union bounds cannot handle cases with infinitely-sized codec classes due to their dependence on $\abs{\mathcal{H}}$, which may occur if we aim to select an encoding function with a continuous parameter.  The following theorem shows that EMD bounds may also be superior for sufficiently large sample sizes. 

\begin{theorem} The following performance results about the Hoeffding-Union bound of \cref{thm:hoeffding-union}  hold:
\begin{enumerate}
    \item The Hoeffding bound will always outperform the McDiarmid finite-sample EMD bound of Theorem \ref{thm:finite-sample-emd} if:
    \[\abs{\mathcal{H}} \leq \left(\frac{2}{\delta}\right)^8\]
    \item It always outperforms the asymptotic EMD bound of Theorem \ref{thm:asymptotic-emd} if:
    \[\abs{\mathcal{H}} \leq \left(\frac{2}{\delta}\right)^{\sigma^2\left(2 + 2 \sqrt{2}\right)^2} \cdot \frac{\delta}{3}\]
    where $\sigma$ is defined as is in Theorem \ref{thm:asymptotic-emd}.
\end{enumerate}
\end{theorem}

\proof{}
To see this, note that the EMD term of the EMD bound is always non-negative, and thus if its McDiarmid term exceeds the Hoeffding term of the union bound, the Hoeffding union bound is superior.
\begin{enumerate}
    \item The McDiarmid term of the finite-sample EMD bound is larger when:
    \[\sqrt{\frac{\ln\left(\frac{2\abs{\mathcal{H}}}{\delta}\right)}{2m}} \leq 3 \sqrt{\frac{\ln \left( \frac{2}{\delta}\right)}{2m}} \]
    Solving for $\abs{\mathcal{H}}$ obtains the desired result.
    \item Here, McDiarmid term of the asymptotic EMD bound is always greater when:
    \[\sqrt{\frac{\ln\left(\frac{2\abs{\mathcal{H}}}{\delta}\right)}{2m}} \leq \sigma(2 + 2\sqrt{2}) \sqrt{\frac{\ln\left( \frac{3}{\delta}\right)}{2m}}\]
    Again, a simple algebraic manipulation is sufficient to obtain the desired inequality.
\end{enumerate}

\subsection{Codecs and Compression}
A \textit{codec} --- a combination of the words ``coder'' and ``decoder'' --- is an algorithm that encodes a sequence to a different data format and decodes it back. Codecs are applicable to domains like audio, video, and they cover a variety of purposes:

\begin{itemize}
    \item \textit{Compression} codecs reduce the memory needed to store a sequence while maintaining most of its the its data. We evaluate the success of a compression scheme metrics such as the amount of memory needed to store compressed sequences and the similarity between the original sequence and its decompressed counterpart. \textit{Lossless} compression schemes perfectly recall the original sequences when decoded, and \textit{lossy} schemes create encodings that require less memory while losing some of the original sequence's data.
    \item \textit{Error-correcting} codecs add redundancy to a sequence to protect its contents in the event that bits are corrupted. Effective error-correcting encoding schemes have relatively small encodings and minimize damage to the original data when subject to random noise.
    \item \textit{Encryption} codecs obfuscate a sequence to prevent a malicious adversary from obtaining an unobfuscated record.  Creators of encryption schemes aim to minimize the probability that the original sequence can be obtained from the obfuscated form, while ensuring that the algorithm quickly and without intensive data needs.
\end{itemize}

While we apply our algorithms to audio compression codecs in this paper, our codec selection algorithms accommodate other types of codecs with different success criteria. We specifically compare various MP3 encoding schemes, which are lossy compression algorithms that vary dramatically in the amount of data needed for compressed files.

\subsection{Codec Selection Problem}

The codec selection problem involves choosing an encoding scheme for a domain like audio, video, or images that meets certain criteria and maximizes an objective dependent on those criteria. Specifically, we examine the audio domain and create a sampling procedure that determines which audio compression scheme performs best with respect to measurements like the size of the compressed data sequences and the similarity between the original sequence and the decompressed version. 

This problem is significant because an optimal compression algorithm should produce a similar sequence to the original and compress to a very small size. These two goals are at odds. All lossless compression schemes are fundamentally limited in compressibility by entropy bounds; therefore, no algorithm can strictly dominate this optimal compression scheme with respect to both decompressed similarity and compression ratio.  Furthermore, the more approximation that is allowed, the better the possible rate of compression, but quality is negatively impacted as well.  We conclude that ``there is no free lunch'' with compression algorithms; a user's relative preferences for fidelity, compressibility, and other measures favor different compression algorithms.

\cite{gupta} examined this kind of problem from a PAC-learning approach. They bounded the performance of different algorithms using bounds based on pseudo-dimension, which generalizes the VC-dimension and associated bounds from \emph{binary classification} to \emph{regression}.

\begin{theorem}[Pseudodimension Bounds]

Let $\HC$ be a finite class of functions with domain $\mathcal{X}$ and range in $[0,H]$. Then, there exists constant $c$ such that for all distributions $\mathcal{D}$ over $\mathcal{X}$, $\delta \in (0,1]$, and $m \in \mathcal{Z}^+$:
\[ \prob[x_1, \dots, x_m \leftarrow \mathcal{D}]{\abs{\frac{1}{m}\sum_{i=1}^m h(x_i) - \expectation[x \distributed \mathcal{D}]{h(x)}} < H \sqrt{\frac{c}{m} \ln \left(\frac{\abs{\HC}^{\frac{1}{\ln 2}}}{\delta} \right)}} \geq 1 - \delta \]

\end{theorem}

We build on their model by replacing their distribution-free pseudo-dimension bounds with data-dependent EMD bounds. Rather than simply finding a function that outperforms others with high confidence, we further expanded their model by also seeking functions that must meet certain constraints. For example, in the audio domain, we might seek a compression algorithm that minimizes the amount of memory needed for the compression sequence while requiring that a compressed and decompressed sequence is sufficiently similar to the original sequence.

\subsection{Perceptual Audio Models}
\label{perceptual}

We apply our framework to the the specific codec selection problem of audio compression, which requires that we select a notion of reconstruction error. While a normalized square error is mathematically convenient and sensible for comparing two vectors, other measures of distance are better tailored to audio data. Because audio files are almost always intended for human listening, we redefine success for denoising algorithms to be how much an decompressed sequences differs from the original sequence \textit{according to a human listener}. To measure this, we employ the techniques of \textit{perceptual audio models}, which attempt to measure how clearly humans hear different sounds.

A central idea in perceptual models is \textit{noise masking}, which occurs when the clean audio signals dominate the added noise from the algorithm in terms of human perception; that is, listeners only hear the certain strong frequencies that drown out other similar frequencies \citep{jayant}. Some perceptual models are \textit{subjective}, meaning that they measure audio quality based on human ratings of sound quality, where listeners assess how much a modified signal resembles the reference signals. Others, like PEAQ (perceptual evaluation of audio quality) are \textit{objective}, and are based more directly on psychoacoustic principles without relying on human-supplied data\citep{thiede}.

In this paper, we use an objective measure for how much compression perturbs each given sound. This method serves as a stand-in for human listeners; we could replace PEAQ with a human listener who rates two sounds by how similar they sound between 0 and 1. Given a better measurement of perceptual audio quality than PEAQ --- whether algorithmic or human-dependent --- we could incorporate it in our algorithms as a quantity to optimize; our algorithm adapts to the provided perceptual quality metric.

\section{A Novel Scale-Sensitive Uniform Convergence Theorem for Unbounded Random Variables}
\label{sec:asymptotic}

While \cref{thm:finite-sample-emd} is appropriate when the range of values is known, this is not always the case, and even when it is known, the result is quite loose when most of the probability mass is concentrated in a relatively small portion of the range.  We now present a scale-sensitive asymptotic uniform convergence bound that assumes only finite variance of each $h \in \HC$.  Our bound is based on analysis of the variance of the EMD.  Computing variances of suprema of empirical processses is a rather subtle topic (the interested reader is referred to \citep{boucheron2013concentration}), but for our purposes, it suffices to understand the \emph{weak variance}, which comes out to
\[
\frac{1}{m}\Expwrt{\bm{x}}{\sup_{h \in \HC}\frac{1}{m}\sum_{i=1}^m(h(\bm{x}_i) - \Expwrt{x' \distributed \rand{D}}{h(x')})^2} \enspace,
\]
and the \emph{wimpy variance}, which commutes the supremum, and comes out to
\[
\frac{1}{m}\sup_{h \in \HC} \Varwrt{x' \distributed \rand{D}}{h(x')} \enspace.
\]
By the Efron-Stein inequality, their sum upper-bounds the variance of the EMD, and both have the same plugin estimator, which we use to obtain asymptotic bounds without a priori variance knowledge.

\begin{restatable}[Asymptotic EMD Uniform Convergence Bounds]{theorem}{thmasymptoticemd}
\label{thm:asymptotic-emd}
Suppose $\HC \subseteq \mathcal{X} \to \R$ with $\abs{\HC} < \infty$, and that $\Varwrt{x}{h(x)}$ exists $\forall h \in \HC$, and take $\bm{x} \distributed \rand{D}^m$.  Now, take either true variance bound
\[
\sigma^2 \doteq \Expwrt{\bm{x}}{\sup_{h \in \HC}\frac{1}{m}\sum_{i=1}^m(h(\bm{x}_i) - \Expwrt{x' \distributed \rand{D}}{h(x')})^2} + \sup_{h \in \HC} \Varwrt{x' \distributed \rand{D}}{h(x')} \leadsto 2\sup_{h \in \HC} \Varwrt{x' \distributed \rand{D}}{h(x')} \enspace,
\]
\emph{or} plugin variance bound estimate:
\[
\sigma^2 \doteq 2\sup_{h \in \HC} \Evarp{h(\bm{x})} \enspace.
\]
The following then hold (asymptotically w.r.t. sample size $m$):

\begin{enumerate}
\item \label{thm:asymptotic-emd:gaussian} $\displaystyle\sup_{h \in \HC} \Expwrt{x' \distributed \rand{D}}{h(x')} - \hat{\Expop}[h(\bm{x})]$ and $\EMD{m}{\HC}{\bm{x}}$ $\leadsto_m$ Gaussian with variance $\leq \frac{2}{m}\sigma^2$.
\item \label{thm:asymptotic-emd:1tail} $\displaystyle\Probop\left(\sup_{h \in \HC} \Expwrt{x' \distributed \rand{D}}{h(x')} - \Eexpp{h(\bm{x})} \geq \sqrt{2}\EMD{m}{\HC}{\bm{x}} + \sigma (2+2\sqrt{2}) \sqrt{\frac{\log(\frac{2}{\delta})}{2m}}\right) \lesssim_m \delta$.
\item \label{thm:asymptotic-emd:2tail} $\displaystyle\Probop\left(\sup_{h \in \HC} \abs{\Expwrt{x' \distributed \rand{D}}{h(x')} - \Eexpp{h(\bm{x})}} \geq \sqrt{2}\EMD{m}{\HC}{\bm{x}} + \sigma (2+2\sqrt{2}) \sqrt{\frac{\log(\frac{3}{\delta})}{2m}}\right) \lesssim_m \delta$.
\end{enumerate}
\end{restatable}

Here \ref{thm:asymptotic-emd:gaussian} shows that the quantities of interest (the supremum deviation and the EMD) are asymptotically Gaussian, and computes their variance.  This is not entirely obvious, as due to the suprema in both quantities, the classical central limit theorem does not apply.  \ref{thm:asymptotic-emd:1tail} then exploits this asymptotic normality to obtain a 1-tailed asymptotic confidence interval via Gaussian Chernoff bounds, and \ref{thm:asymptotic-emd:2tail} is the 2-tailed bound, which requires a confidence interval wider by a factor of $\sqrt{{\log(\frac{3}{\delta})} / {\log(\frac{2}{\delta})}}$.

This result should be compared with the finite-sample bounded-difference results of \cref{thm:finite-sample-emd}.  In that case, $\sigma^2 \leq 1$ (due to boundedness).  To isolate the effects of changing the concentration inequality from the effects of the $2$ factor being asymptotically replaced by a $\sqrt{2}$, we use $2\EMD{m}{\HC}{\bm{x}}$ instead of $\sqrt{2}\EMD{m}{\HC}{\bm{x}}$, which yields a concentration inequality term of $\sigma6\sqrt{2}\sqrt{\frac{\log(\frac{2}{\delta})}{2m}}$, as opposed to the $\sigma3\sqrt{\frac{\log(\frac{1}{\delta})}{2m}}$.  Thus here, without making a boundedness assumption, we obtain a result that differs only in the constants.  Often $\sigma^2 \ll 1$, in which case the asymptotic bound is tighter, and the asymptotic result also applies in situations where values are unbounded, thus it is applicable in a far broader domain.

\section{Codec Selection Algorithms}

We present algorithms that determine the optimal codec in $\mathcal{H}$ that satisfies the constraints. Both algorithms rely on placing probabilistic bounds on the means for each criterion to obtain intervals where each mean lie with high probability. If an intervals lies within the constraint space, we conclude that the corresponding function is satisfactory. If the confidence intervals for two objectives do not overlap, then we can conclude that the greater objective is less optimal with high probability. 

The Global Sampling (GS) algorithm in Section \ref{global-sampling} creates those intervals by encoding all samples with each function. The Progressive Sampling with Pruning (PSP) algorithm in Section \ref{psp} differs by finding the intervals with only some of the samples and disqualifying suboptimal or unsatisfactory functions in an online manner. The two algorithms use the same bounds to construct confidence intervals, but they take different sampling approaches. The GS algorithm processes \textit{non-adaptively} because it encodes and decodes all samples regardless. The PSP eliminates codecs that are sufficiently unlikely of being optimal and chooses samples \textit{adaptively} from remaining codecs. Adaptive selection has the advantage of avoiding wasted time o samples with clearly sub-optimal codecs.

These algorithms differ from multi-armed bandit approaches because they apply every sample to each codec. Bandit-based algorithms must draw each sample i.i.d., which means that samples cannot be reused on different codecs and requires a much larger data source. Moreover, bandit bounds do not take advantage of uniform convergence properties, so we expect them to perform worse than our algorithms with EMD bounds when there is a large class of codecs.

First, we introduce the notation for the model that is used for both algorithms.

\subsection{Problem Formulation}

For our model, the encoding schemes are represented as a class of functions $\mathcal{H}$ with $h: \mathcal{X} \rightarrow \mathcal{Y}$ where $\mathcal{X}$ is the set of input objects to encode and $\mathcal{Y}$ is a representation of the encoding along with useful metadata about the transformation. We also have a set of criteria $\mathcal{C}$ that describes the performance of the function with $c: \mathcal{X} \times \mathcal{Y} \rightarrow [0,1]$ for $c \in \mathcal{C}$. In the audio compression domain, $\mathcal{H}$ contains compression schemes like MP3, while $\mathcal{C}$ contains measurements like the ratio of the size of the original data sequence to the size of the compressed sequence. 

We measure the success of an encoding with an objective $V: \mathcal{X} \times \mathcal{Y} \rightarrow \mathbb{R}$ such that $V(x, y)$ is a linear combination of $c(x, y)$ for $c \in \mathcal{C}$. Thus, we can also express $V$ as a function of a vector of criterion values: $V(u)$ for where each element of $u$ corresponds to some $c(x, h(x))$ for $c \in \mathcal{C}$. We also seek functions that satisfy some set convex combination of linear inequality constraints in terms of the criteria: let $\mathcal{W} \subseteq \mathcal{Y}$ be the region of the function's range such that those constraints are satisfied. Thus, we seek $h^* \in \mathcal{H}$ such that with high probability for $x$ drawn from $\mathcal{X}$ over some distribution and for all $h \in \mathcal{H}$, $V(x, h^*(x)) \leq V(x,h(x)) + \epsilon$ and $h^*(x) \in \mathcal{W}$. Because two functions could be arbitrarily similar, we can only guarantee that our function’s objective mean is no more than $\epsilon$ worse than the optimal with high probability.

For simplicity, define $c \circ h$ such that $(c \circ h)(x) = c(x,h(x))$ for $c \in \mathcal{C}, h \in \mathcal{H}$. In addition, we let $c \circ \mathcal{H} = \{c \circ h : h \in \mathcal{H}\}$ represent a function class for corresponding to each criterion.

\subsection{Global Sampling Algorithm}\label{global-sampling}

The global sampling algorithm encodes every sample in $S \subseteq \mathcal{X}$ with each codec in $\mathcal{H}$. It computes an empirical mean $\hat{e}_{h,c}$ and a confidence interval $\hat{E}_{h,c}$ for the means of each criterion $c \in \mathcal{C}$ with each codec $h \in \mathcal{H}$. Based on those bounds, we can determine which hypotheses reside within the constraint space with high probability, and we also find similar means and bounds for the objective for each codec.  We assume an objective function $V$ that is a nonnegative weighted sum of criteria (occasionally writing $V \cdot \bm{c}$ for criteria vector to specify this), and estimate the optimal $h \in \HC$ w.r.t. objective $V$.  We also allow for a constraint space $\mathcal{W}$, in which case we estimate the optimal codec known to satisfy $\HC$, as well as some auxilliary information used to quantify the quality of the estimation.

When assessing a whether a codec meets the specified constraints with confidence on a set of samples, there are three possible conclusions: with confidence, the codec meets the constraints; with confidence, the codec does \textit{not} meet the constraints; or neither conclusion can be mode confidently. This poses a question; if the algorithm is inconclusive about the constraint satisfiability of the codec with the smallest empirical mean, should report that codec as the optimal scheme? This motivates a need for \textit{liberal} and \textit{conservative} selections, where codecs with undetermined constraint satisfaction are eligible to be optimal according to liberal selection, but not for conservative. This distinction is further necessary because some criterion's satisfiability may be impossible to verify. If its true mean lies on the constraint boundary, the confidence interval will never lie entirely inside or outside of the feasible region, and confidence is impossible regardless of the number of samples.

\begin{algorithm}
    \caption{\textsc{GlobalSampling}$(S, \mathcal{H}, \mathcal{C}, V, \mathcal{W}, \delta)$}
    \begin{algorithmic}[1]
        \State{\textbf{Input:} samples $S$, codec family $\HC$, criteria $\mathcal{C}$, objective $V$, constraint space $\mathcal{W}$, failure probability $\delta \in (0,1)$}
        \State{\textbf{Output:} liberal and conservative codec candidate sets $\hat{h}_L, \hat{h}_C \subseteq \mathcal{H}$, criteria estimates $\hat{e}_{\hat{h},c}$, criteria error bounds $\bm{\epsilon}$}
        \For{$c \in \mathcal{C}, h \in \mathcal{H}$} 
            \State{$\hat{e}_{h, c} \gets \frac{1}{|S|} \sum_{x \in S} c(x, h(x))$}
            \State{$d_c \gets \EMD{\abs{S}}{c \circ \HC}{S}$}
            \State{$\eta \gets 3\sqrt{\ln(2 |\mathcal{C}| / \delta) / 2|S|}$}
            \State{$\bm{\epsilon}_c \gets 2d_c + \eta$}
            \State $\hat{E}_{h, c} \gets [\hat{e}_{h,c} - \bm{\epsilon}_c, \hat{e}_{h,c} + \bm{\epsilon}_c] \cap [0, 1]$
        \EndFor
        \State{$\hat{\HC}_L \gets \{h \in \HC : \hat{E}_{h, \cdot} \cap \mathcal{W} \neq \emptyset\}$} \Comment{Set of all codecs that may satisfy $\mathcal{W}$}
        \State{$\hat{\HC}_C \gets \{h \in \HC : \hat{E}_{h, \cdot}  \subseteq \mathcal{W}\}$} \Comment{Set of all codecs that satisfy $\mathcal{W}$ w.h.p.}

        \State{$\hat{h}_L \gets \{h \in \hat{\HC}_L : \inf_{c \in \hat{E}_{h,.}} V(c) \leq \inf_{h' \in \hat{\HC}_L} \sup_{c \in \hat{E}_{h', \cdot}} V(c)\}$} \Comment{Select near-optimal liberal codecs}
        \State{$\hat{h}_C \gets \{h \in \hat{\HC}_C : \inf_{c \in \hat{E}_{h,.}} V(c) \leq \inf_{h' \in \hat{\HC}_C} \sup_{c \in \hat{E}_{h', \cdot}} V(c)\}$} \Comment{Select near-optimal conservative codecs}
        \State{\Return{$(\hat{h}_L,\hat{h}_C,\hat{e},\bm{\epsilon})$}}
        
    \end{algorithmic}
\end{algorithm}

We obtain basic theoretical guarantees about the outcomes of the algorithm. Note that the algorithm uses the EMD bounds from Theorem \ref{thm:finite-sample-emd} to define the confidence intervals. Those bounds can be replaced with the bounds for confidence intervals without affecting the rest of the algorithm or the theoretical results of Theorem \ref{thm:gs-guarantees}.

\begin{theorem}
\label{thm:gs-guarantees}
    Suppose we run $\textsc{GlobalSampling}(S, \mathcal{H}, \mathcal{C}, V, \mathcal{W}, \delta)$ and obtain $(\hat{h}_L, \hat{h}_C, \hat{e}_{\cdot, \cdot}, \bm{\epsilon})$, with $S$ drawn i.i.d. from $\rand{D}$.  Take $\HC_{\mathcal{W}} = \{h \in \HC : \Expwrt{x \distributed \rand{D}}{\mathcal{C}(x, h(x))} \in \mathcal{W}\}$ to be the subset of $\HC$ that satisfies $\mathcal{W}$ in expectation.  The following hold (simultaneously) with probability at least $1 - \delta$ over choice of $S$:
    \begin{enumerate}
        \item \label{thm:gs-guarantees:liberal} If $\HC_{\mathcal{W}}$ is nonempty, then $\hat{h}_L$ is nonempty.
        \item \label{thm:gs-guarantees:conservative} The set of conservative hypotheses satisfies $\hat{h}_C \subseteq \HC_{\mathcal{W}}$.
        \item \label{thm:gs-guarantees:criterion-approximation} Each estimated mean criterion value for each codec $c \in \mathcal{C}$ is approximately correct, satisfying
        \[
        \sup_{h \in \HC} \abs{\Expwrt{x \distributed \rand{D}}{c(x, h(x))} - \hat{e}_{h,c}} \leq \bm{\epsilon}_c \enspace.
        \]
        \item \label{thm:gs-guarantees:objective-approximation} Each estimated objective value is approximately correct, satisfying
        \[
        \sup_{h \in \HC} \abs{\Expwrt{x \distributed \rand{D}}{V(\mathcal{C}(x, h(x)))} - \hat{e}_{h,c}} \leq V \cdot \bm{\epsilon} \enspace.
        \]
        \item \label{thm:gs-guarantees:optimality} Assume $\hat{h}_L$ and $\hat{h}_C$ are nonempty for the lower and upper bounds, respectively. We then bound the true objective of the optimal $h^* \in \HC_{\mathcal{W}}$ in terms of our empirical estimates as
        \[
        \inf_{h \in \hat{h}_L} V(\hat{e}_{h,\cdot} - \bm{\epsilon}) \leq \inf_{h^* \in \HC_\mathcal{W}} \Expwrt{x \distributed \rand{D}}{V(\mathcal{C}(x, h^*(x)))} \leq \inf_{h \in \hat{h}_C} V(\hat{e}_{h,\cdot} + \bm{\epsilon}) \enspace.
        \]
        
    \end{enumerate}
\end{theorem}

\begin{proof}
To show each part of the theorem, we first note that the conclusion of \cref{thm:finite-sample-emd} holds by union bound for each criterion with probability at least $1 - \delta$.  This immediately implies \ref{thm:gs-guarantees:criterion-approximation}, which itself directly implies \ref{thm:gs-guarantees:objective-approximation}.

To see \ref{thm:gs-guarantees:liberal}, consider that the confidence intervals about the empirical estimate of each criterion for each hypothesis in $\HC$ contain their true values.  Consequently, any codec that satisfies $\mathcal{W}$ will be a member of $\hat{\HC}_L$, and thus $\hat{h}$ will be nonempty.  Similarly, to see \ref{thm:gs-guarantees:conservative}, $\hat{\HC}_C$ is the set of hypothesis such that any possible configuration of criterion values (within their confidence intervals) satisfy $\mathcal{W}$, thus all conservative codecs (members of $\hat{h}_C$) must satisfy $\mathcal{W}$.

Finally, to see \ref{thm:gs-guarantees:optimality}, note that the set $\HC \setminus \hat{\HC}_L$ contains only hypotheses that are known not to satisfy $\mathcal{W}$ (with probability at least $1 - \delta$), thus we may ignore them, and consider only $\hat{\HC}_L$ when lower-bounding the true optimal objective.  Similarly, when upper-bounding the true objective, we must consider only codecs known to satisfy $\mathcal{W}$ (as a codec that does not satisfy $\mathcal{W}$ could have a lower objective), thus we consider $\hat{h}_C$ instead of $\hat{h}_L$.
\end{proof}

\ref{thm:gs-guarantees:liberal} characterizes situations in which we will obtain a liberal codec estimate, namely, so long as a valid codec exists, we will obtain a liberal estimate with high probability.   \ref{thm:gs-guarantees:conservative} tells us that the true criterion values associated with conservative hypotheses always satisfy the constraints $\mathcal{W}$: while it is a simple matter to guarantee that empirical criteria estimates satisfy $\mathcal{W}$, this is subtler but more useful property.  \ref{thm:gs-guarantees:criterion-approximation} and \ref{thm:gs-guarantees:objective-approximation} characterize the quality of our estimated criteria and objective values (showing they uniformly approximate the true values).  Finally \ref{thm:gs-guarantees:objective-approximation} characterizes the true optimal objective value of all codecs that satisfy $\mathcal{W}$, which is rather subtle, as we accomplish this without actually computing the optimal codec or knowing exactly which codecs satisfy $\mathcal{W}$.

\subsection{Progressive Sampling with Pruning Algorithm}\label{psp}

Our \emph{Progressive Sampling with Pruning} (PSP) procedure samples data sequences to the codec scheme from the distribution of sequences in batches of increasing size. It is based on the progressive sampling method introduced by \citep{elomaa2002progressive}, and later applied to unsupervised learning by \citep{riondato2015mining,riondato2016abra}. Our technique differs in that we use the results of the progressive sampling to eliminate from consideration, or prune, functions whose results will be insufficient with high confidence. Based on the results of each batch, we empirically estimate the means of $c(x,h(x))$ for criterion $c$ and hypothesis $h$ over samples $x$. We bound the means with high probability with Rademacher complexity and prune functions which (1) with high confidence lie outside of $\mathcal{W}$ or (2) with high confidence has a greater mean value of $V(x,h(x))$ than $V(x,h'(x))$ for some other $f'$ that is in $\mathcal{W}$ with high confidence.

The key advantage of pruning is that the statistical power of the technique is only minimally impacted by poorly performing functions, as these are quickly pruned, and additionally minimal computation time is spent on these functions.  Intuitively, this is similar to the idea of \emph{local Rademacher complexity}, where localized function families \citep{bartlett2005local} are centered around the optimal function in a family.

The algorithm is parameterized by $\delta$ and $\epsilon$, whose meanings are analogous to those in PAC-Learning framework \citep{valiant}. $\delta$ represents our level of certainty in each bound. That is, we require that each bound holds with probability $1 - \delta$. $\epsilon$ represents the tightness of our bound. That is if the algorithm terminates, we guarantee an $\epsilon$-optimal codec with probability at least $1 - \delta$. We choose any values of $\delta$ and $\epsilon$, and the algorithm uses no more than twice as many samples as the minimum needed for sufficiently tight EMD bounds.

We repeat this process over batches that double in size after each iteration. Because each batch has more samples than the previous one, it can place tighter bounds than the preceding ones can, thus allowing it to be more confident in its empirical means. We then disqualify more functions because we conclude with confidence that certain functions outperform other functions due to the shrinking confidence intervals.  This means that fewer functions need to be run on each subsequent batch of inputs, which limits the number of unnecessary encoding steps. The algorithm terminates when there is exactly one remaining function that satisfies the constraints, when it is shown that no function satisfies the constraints, or when no more samples are available. Therefore, the algorithm finds the function with the smallest objective subject to the constraints with high confidence. We give pseudo-code for the algorithm in an attached figure:

\begin{algorithm}
    \caption{PSP$(S, s_0, \mathcal{H}, \mathcal{C}, V, \mathcal{W}, \epsilon, \delta)$}
    \begin{algorithmic}[1]
        \State{\textbf{Input:} samples $S$, initial batch size $s_0$, codec class $\mathcal{H}$, criterion set $\mathcal{C}$, objective $V$, constraint space $\mathcal{W}$, confidence $\epsilon \in \{0,1\}$, failure probability $\delta\in \{0,1\}$}
        \State{\textbf{Output:} liberal and conservative codec sets $\hat{h}_L, \hat{h}_C \subseteq \mathcal{H}$, empirical criteria estimates $\hat{e}_{\hat{h},c}$, criteria confidence intervals $\hat{E}_{\hat{h},c}$}
        \State  $\hat{E}_{h, c} \gets [0,1]$ for all $h \in \mathcal{H}, c \in \mathcal{C}$
        \State $n \gets \lfloor \log_2(\frac{|S|}{s_0} + 1)\rfloor$ \Comment{Maximum number of iterations}
        \State $\HC_1 \gets \HC$ \Comment{Begin by considering all codecs}
        \For{$ i \in \{1, \dots, n\}$}
            \State Let $S_i$ be $2^{i-1} \cdot s_0$ unused samples from $S$
            \State $\eta_i \gets 3\sqrt{\ln(2n |\mathcal{C}| / \delta) / 2|S_i|}$ \Comment{McDiarmid's inequality term}
            \For{$c \in \mathcal{C}, h \in \mathcal{H}$}
                \State $\hat{e}_{h, c} \gets \frac{1}{|S_i|} \sum_{x \in S_i} c(x,h(x))$
                \State $d_{c,i} \gets \emd(c \circ \mathcal{H}, S_i)$
                \State $\hat{E}_{h, c} \gets \hat{E}_{h,c} \cap [\hat{e}_{h,c} - 2 d_{c,i} - \eta_i, \hat{e}_{h,c} + 2 d_{c,i} + \eta_i]$ \Comment{Compute EMD confidence intervals}
            \EndFor

            \State{$\hat{h}_L \gets \{h \in \HC_i : \inf_{c \in \hat{E}_{h,.}} V(c) \leq \inf_{h' \in \hat{\HC}_L} \sup_{c \in \hat{E}_{h', \cdot}} V(c) \wedge E_{h,\cdot} \cap \mathcal{W} \neq \emptyset \}$}
            \State{$\hat{h}_C \gets \{h \in \HC_i : \inf_{c \in \hat{E}_{h,.}} V(c) \leq \inf_{h' \in \hat{\HC}_C} \sup_{c \in \hat{E}_{h', \cdot}} V(c) \wedge E_{h,\cdot} \subseteq \mathcal{W}\}$}

            \State{$\HC_{i+1} \gets \{h \in \HC_{i} : \hat{E}_{h, \cdot} \cap \mathcal{W} \neq \emptyset\}$} \Comment{Prune codecs that provably violate $\mathcal{W}$}
            \If{$\hat{h}_C \neq \emptyset$}
                \State{$\HC_{i+1} \gets \{h \in \HC_{i+1} : \inf_{c \in \hat{E}_{h,\cdot}}V(c) \leq \inf_{h' \in \hat{h}_C} \sup_{c \in \hat{E}_{h',\cdot}} V(c)\}$} \Comment{Prune suboptimal codecs}
            \EndIf
            
            \If{$\lvert\HC_{i+1}\rvert = 1 \vee \inf_{h \in \hat{h}_C}\sup_{c \in \hat{E}_{h,\cdot}} V(c) \leq \inf_{h' \in \hat{h}_L}\inf_{c \in \hat{E}_{h',\cdot}} V(c) + \epsilon \vee i = n$}
                \State{\Return{$(\hat{h}_L, \hat{h}_C, \hat{e}, \hat{E})$}} 
                \Comment{Terminate if 1 codec remains, $\epsilon$-optimality is reached, or $S$ is exhausted}
            \EndIf
        \EndFor
    \end{algorithmic}
\end{algorithm}

Based on the algorithm, we obtain analogous guarantees to \cref{thm:gs-guarantees}.

\begin{restatable}[Guarantees of the PSP Algorithm]{theorem}{thmpsp}
\label{thm:psp-guarantees}
    Suppose we run $\text{PSP}(S, s_0, \mathcal{H}, \mathcal{C}, V, \mathcal{W}, \epsilon, \delta)$, for $S$ drawn i.i.d. from $\rand{D}$ and obtain $(\hat{h}_L, \hat{h}_C, \hat{e}_{\hat{h},.}, \hat{E}_{\hat{h},.}, \epsilon)$. Take $\HC_{\mathcal{W}} = \{h \in \HC : \Expwrt{x \distributed \rand{D}}{\mathcal{C}(x, h(x))} \in \mathcal{W}\}$ to be the subset of $\HC$ that satisfies $\mathcal{W}$ in expectation.  Then, the following hold with probability $1 - \delta$ (simultaneously) over choice of $S$:
    \begin{enumerate}
        \item \label{thm:psp-guarantees:liberal} If $\HC_{\mathcal{W}}$ is nonempty, then $\hat{h}_L$ is nonempty.
        \item \label{thm:psp-guarantees:conservative} Each conservative hypothesis $h \in \hat{h}_C$ satisfies $h \in \HC_{\mathcal{W}}$.
        
        \item \label{thm:psp-guarantees:criterion-approximation} The true mean of each criteria and codec lies within our confidence rectangle, satisfying
        \[
        \expectation[x]{c(x,h(x))} \in \hat{E}_{h, c} \ \forall c \in \mathcal{C}, h \in \HC \enspace.
        \]
        \item \label{thm:psp-guarantees:objective-approximation} The true objective of each codec lies within our confidence interval for the objective, satisfying
        \[
        \expectation[x]{V(\mathcal{C}(x,h(x)))} \in V(\hat{E}_{h,.}) \ \forall h \in \HC \enspace.
        \]
        
        \item \label{thm:psp-guarantees:optimality} Assume $\hat{h}_L$ and $\hat{h}_C$ are nonempty for the lower and upper bounds, respectively. We then bound the true objective of the optimal $h^* \in \HC_{\mathcal{W}}$ in terms of our empirical estimates as
        \[
        \inf_{h \in \hat{h}_L} \min(V(\hat{E}_{h,\cdot})) \leq \inf_{h^* \in \HC_\mathcal{W}} \Expwrt{x \distributed \rand{D}}{V(\mathcal{C}(x, h^*(x)))} \leq \inf_{h \in \hat{h}_C} \max(V(\hat{E}_{h,\cdot})) \enspace.
        \]
    \end{enumerate}
\end{restatable}
\begin{proof}
Proof of this result is essentially identical to that for global sampling, except that we now require the EMD bounds to hold simultaneously (by union bound) for each iteration of the progressive sampling.  Note that every pruned function provably violates $\mathcal{W}$ or is provably suboptimal, thus removing them never removes a candidate for $\hat{h}_L$ or $\hat{h}_C$, and indeed only allows us to obtain tighter bounds for each criterion, refining our final estimated values and selected functions.
\end{proof}

While the two algorithms are structured similarly, with additive approximation $\epsilon$, error probability $\delta$, and sample size $m$, each can be modified slightly to obtain a sufficiently large value for one given values for the other two. 
\begin{itemize}
    \item The global sampling algorithms naturally works with a fixed set of data points where the accuracy can be determined: given $m$ and $\delta$, this algorithm can easily return a satisfying $\epsilon$ value.
    \item The progressive sampling algorithm is more intuitive when we aim to only use as much data as is necessary. Given $\epsilon$ and $\delta$, PSP uses roughly no more than twice as many samples required for the $\epsilon$-$\delta$ bound, as at this point, its sample size equals that we would use for global sampling.  Often PSP requires significantly less data, as it can terminate early, and pruning can lead to much tighter confidence intervals.
\end{itemize}

\section{Experimental Results}

\subsection{Implementation Details}
We implemented both the Global Sampling algorithm and the Progressing Sampling with Pruning algorithm in Java. Our codebase is sufficiently general to apply then algorithm to a wide range of domains; one simply needs to implement the appropriate codecs and criteria.

We applied the algorithm to the audio codec selection domain. The function class $\mathcal{H}$ corresponds to compression algorithms. For this example, we choose between LAME MP3 encoders of different variable bit-rates --- each one has a rating between V0 and V9 representing the how many bits are used to encode segments of audio \citep{lame}. V0 has the highest bit-rate, which means it features the highest quality sound yet also reduces the file size the least; V9 has the lowest bit-rate and thus has the lowest quality and the smallest output files.

We can also compare variable bit-rate (VBR) codec schema with constant bit-rate (CBR) and average bit-rate (ABR) schema. VBR schemes dynamically choose how many bits to compress, which tends to give the best results because more bits can be used to compress complex segments of audio than simple segments. CBR algorithms use the same amount of bits for each segment, which means that complex segments may lose key sounds, and simple segments may have too much redundancy; however, CBR schema guarantees an exact compression ratio. ABR is a compromise between the two that aims for a certain compression ratio, but can dedicate more resources to compressing complex portions of the audio track, while dedicating fewer resources to easily-compressed simple regions.

We considered several kinds of criteria for $\mathcal{C}$, which are meant to measure qualities of the compression algorithms that users would care about:

\begin{itemize}[wide, labelwidth=!, labelindent=0pt]
\setlength{\itemsep}{0pt}
\setlength{\parskip}{0pt}
    \item \textit{PEAQ Objective Difference} ($c_1$) is an \textit{objective perceptual audio model} because it uses on computational models of the ear to how much two audio sequences differ in perception. We discuss these models in Section \ref{perceptual}. We compute PEAQ using the GSTPEAQ codebase created by \cite{holters}. 
    \item \textit{Root Mean Squared Error} ($c_2$) treats the two audio sequences as vectors of numbers in $[-1,1]$ and computes a normalized L2-distance between the two. This criterion examines similarities in the file representation while neglecting the perceptual differences.
    \item \textit{Compression Ratio} ($c_4$) represents the ratio of size of the compressed audio sequence to the size of the original audio sequence. Smaller values indicate that a compression algorithm is effective at reducing the size of a given audio file.
    \item \textit{Compression Time} ($c_5$) is the time in seconds needed for the compression algorithm to compress the audio sequence. Because all criteria must output values in $[0,1]$, we actually take the minimum of the compression time and 1 --- this is a reasonable assumption because all compression schemes we have observed so far take much less time than one second.
    \item \textit{Decompression Time} ($c_6$) is the same as Compression Time, except that it measures the amount of time needed to decompress the file back to WAV format.  

\end{itemize}

\begin{figure}
    \centering
    \includegraphics[width=0.8\textwidth,trim={2cm 0.5cm 3cm 1cm},clip]{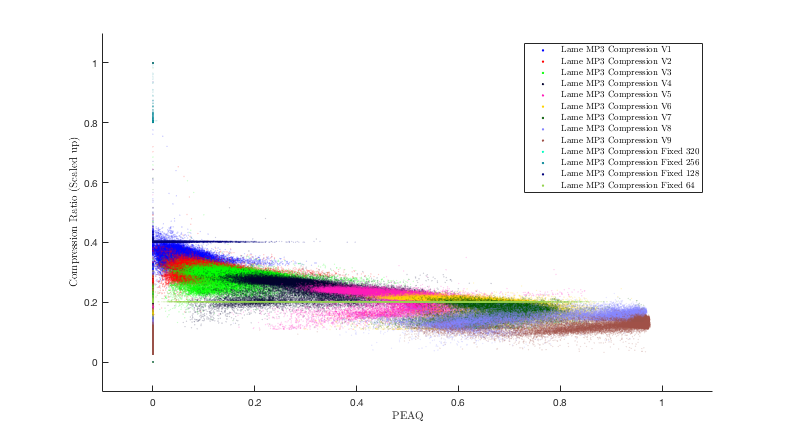}
    \caption{A scatter-plot of scaled compression ratios versus PEAQ divergences for music segments when encoded by a wide range of codecs.}
    \label{peaq_cr}
\end{figure}

To visualize trends in these criteria for these codec functions, we compare PEAQ divergence and compression ratio in Figure \ref{peaq_cr} on music files. Each point represents those criteria values for a sample encoded with a given function. We observe an inverse relationship with respect to compression ratio and PEAQ values. This makes sense because using more bits to encode a file leads to a smaller difference in sound between the original sequence and the decompressed sequence. Because they have fixed bit-rates, the CBR schemes have constant compression ratios, but vary widely in sound fidelity. The most accurate CBR schemes (256 and 320 bits) have zero perceived difference between the input and output files, but they have very high compression ratios. Since we consider objectives that balance PEAQ and CR as linear combinations of the two terms, we can visualize the objective on the plot as the line through the origin defined by vector corresponding to the combination. The optimal scheme will be the one whose mean is closest to the origin when projected onto that line. 

From the criteria introduced above, we construct objectives $V$ that are linear combinations of these criteria. We run the the algorithm with each objective and observe how different compression functions are selected based on our stated preferences. We intuitively want to limit the compression ratio and changes to the sound. Future users could modify this objective to be other combinations of criteria depending on what the listener values in a compression algorithm. Our objectives are as follows:
\begin{itemize}
\item Root Mean Squared Error: $V_1(h(x)) := c_2(h(x))$.
\item PEAQ divergence: $V_2(h(x)) := c_1(h(x))$.
\item Compression ratio: $V_3(h(x)) := c_4(h(x))$.
\item Combination of PEAQ and CR: $V_4(h(x)) := \frac{1}{3} c_1(h(x)) + \frac{2}{3} c_4(h(x))$
\end{itemize}

In addition, we can apply confidence intervals to the \textit{variance} of criterion $c$. This can be useful in the audio domain because we might insist on a particularly dependable encoding function by requiring that the compression ratio have little variation. Because $\Varwrt{z}{c(h(z))} = \expectation{c(h(z))^2} - \expectation{c(h(z))}^2$, we can obtain confidence intervals for $c^2$ and $c$ and combine their intervals. We can then find a suitable confidence interval for the variance with that holds with at least the same probability as the two other confidence intervals:
\[
\hat{E}_{\var{c},h} \subseteq [\min \hat{E}_{c^2,h} - \max \hat{E}_{c,h}^2, \max \hat{E}_{c^2},h -  \min \hat{E}_{c,h}^2] \enspace
\]
Variances are natural criteria to incorporate in constraints as a way to control extreme values. Although we do not demonstrate any constraint examples here, our framework easily permits these tests.

To test the model, we run it on ten-second clips from open source audio books on LibriVox\citep{mcguire}. We collected over 16000 samples in that manner. When running the progressive sampling algorithm, we started by using 25 samples in the first round and doubling that quantity until terminating after ninth samples were used on the eighth round. Notably, our algorithm often fails to isolate the best encoding function with EMD-based bounds --- that would require more data. However, it succeeds in creating tight confidence intervals and in using those intervals to eliminate the worst function given our objective.

The tests that follow demonstrate the versatility of our selection algorithm by providing showing the differences in algorithmic behavior by changing different variables. Because the rigorous algorithm converges slowly for a relatively small number of samples, we set $\delta$ and $\epsilon$ in order to cleanly observe the behavior of the algorithm.
Unless otherwise stated, assume that $\epsilon = 0.05$ and $\delta = 0.01$ for each test.

\subsection{Global Sampling Experiments}
We tested the Global Sampling algorithm to compare the convergence rates of each bound. Table \ref{global-table} contains the results for the algorithm with our Finite-Sample EMD bounds (\cref{thm:finite-sample-emd}), our Asymptotic EMD bounds (\cref{thm:asymptotic-emd}), the standard Hoeffding+Union bound that dominates Roughgarden's pseudo-dimension bounds (\cref{thm:hoeffding-union}), and asymptotic Gaussian-Chernoff bounds (\cref{thm:asymptotic-emd}). We use $\approx$ 16,000 10-second samples from audio books, test on nine variable bit-rate MP3 encoding schemes and four constant bit-rate MP3s, and use only the PEAQ divergence and Compression Ratio criteria. We track convergence by measuring the width of the confidence interval for the objective function of the optimal scheme and how many total compression schemes remain under consideration when the algorithm finishes; the algorithm only concludes that a codec is optimal with at least probability $1 - \delta = 0.99$ if the number of remaining schemes is 1.

\begin{table}

{\scriptsize
\begin{tabular}{|l|l|c|c|c|c|c|c|c|c|}
    \hline
     \multirow{2}{*}{Objective} & \multirow{2}{*}{Optimal Codec}
     & \multicolumn{2}{c|}{Finite-Sample EMD} & \multicolumn{2}{c|}{Asymptotic EMD} & \multicolumn{2}{c|}{Hoeffding Union} & \multicolumn{2}{c|}{\scriptsize Gaussian-Chernoff Union} \\ \cline{3-10}
     & & Width & $\lvert \hat{\HC} \rvert$ & Width & $\lvert \hat{\HC} \rvert$ & Width & $\lvert \hat{\HC} \rvert$ & Width & $\lvert \hat{\HC} \rvert$ \\ \hline
     $V_2$ (PEAQ) & MB3 VBR V1 & 0.082 & 3 & 0.018 & 1 & 0.032 & 1 & 0.007 & 1 \\ \hline
     $V_3$ (CR) & MB3 VBR V9 & 0.081 & 1 & 0.005 & 1 & 0.032 & 1 & 0.006 & 1\\ \hline
     $V_4$ (PEAQ \& CR) & MB3 VBR V5 & 0.082 & 10 & 0.009 & 5 & 0.032 & 10 & 0.003 & 3\\ \hline
\end{tabular}
}
\caption{A comparison of experimental results when the Global Sampling algorithm is applied to objectives $V_2$, $V_3$, and $V_4$ with four types of bounds.  Width denotes the width of each confidence interval, and $\lvert \hat{\HC} \rvert$ denotes the number of members of $\HC$ that GS concludes could possibly be optimal (i.e. their confidence intervals intersect those of the estimated optimal codec).}
\label{global-table}
\end{table}

Note that the approximate bounds converge much faster than the finite-sample bounds, with much smaller intervals than the EMD and Hoeffding bounds, as the latter are based on bounded differences. Because PEAQ divergence and Compression Ratio are at odds (i.e. improvements to audio fidelity require additional data for the compression), it is unsurprising that objective $V_4$ is the most difficult to separate.

\subsection{Progressive Sampling Experiments}
Like with the Global Sampling algorithm, we compare the effectiveness of different bounds using the Progressive Sampling with Pruning algorithm under various objectives. Because the Progressive Sampling approach is iterative, we visualize the results as plots of the confidence intervals of the objectives for each iteration. Each encoding scheme has a series of intervals for each iteration where it is active; once a scheme is pruned, its intervals are no longer displayed. The plot in Figure \ref{peaq-demonstration} demonstrates visually how the algorithm prunes codecs and shrinks interval sizes with the Finite-Sample EMD bounds.

\begin{figure}
	\centering
	\includegraphics[width=\textwidth]{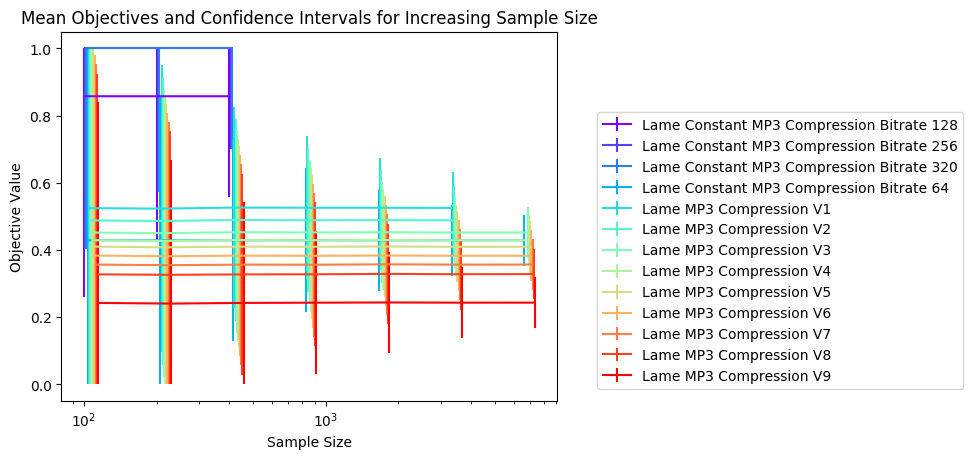}
	\vspace{-0.5cm}
	\caption{A demonstration of the PSP algorithm with objective $V_3$ and Finite-Sample EMD bounds.}
	\label{peaq-demonstration}
\end{figure}

Because the algorithm guarantees the true objective mean of the selected codec to be no more than $\epsilon$ greater than that of the optimal codec with probability $\delta$, the algorithm may terminate early when the width of the interval of the empirically-optimal scheme is smaller than $\epsilon$. Note that Figure \ref{fig:psp-objs} shows all nine iterations of interval-tightening, regardless of when the algorithm actually terminates.
\begin{figure}

    \vspace{-0.75cm}
    \begin{multicols}{2}

    \includegraphics[width=0.45\textwidth,trim={0 1.7cm 0 0},clip]{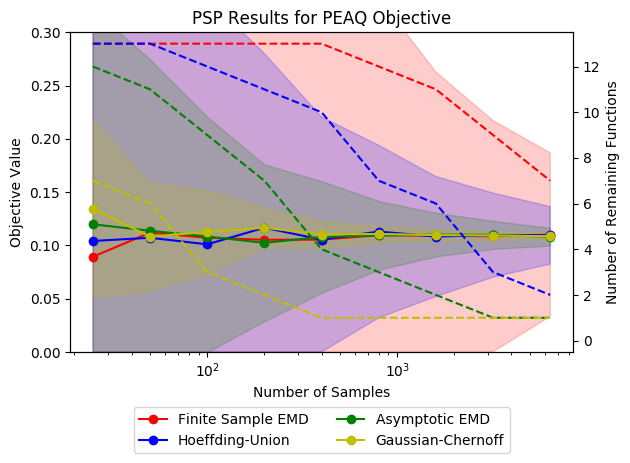}
    \vspace{-0.25cm}
    \includegraphics[width=0.45\textwidth]{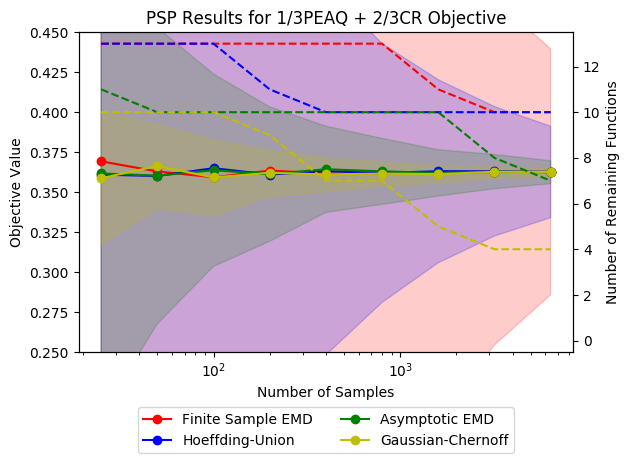}
	\includegraphics[width=0.45\textwidth,trim={0 1.7cm 0 0},clip]{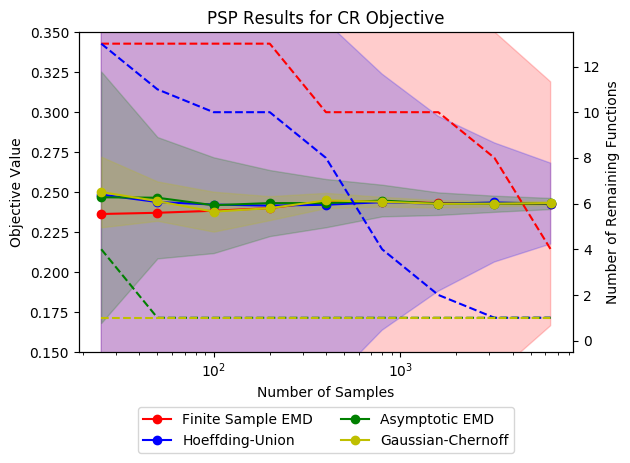}

	\caption{The results of PSP algorithm with objective $V_2$ (PEAQ divergence, top left) and $V_3$ (Compression Ratio, top right), and $V_4$ (PEAQ and Compression Ratio, bottom left) with the all four bounds. The objective confidence intervals are shown as semitransparent solid regions, and the estimated mean objective values are plotted as solid lines. The number of codecs remaining after each progressive sampling iteration is plotted with dashed lines.}
	\label{fig:psp-objs}
	
    \end{multicols}
\end{figure}

We apply the PSP algorithm to the PEAQ divergence objective in \cref{fig:psp-objs} (top left). Once again, we have 16,000 10-second audio book samples, thirteen variants of VBR and CBR MP3, and only PEAQ divergence and Compression Ratio criteria. The approximate Asymptotic EMD bounds and Guassian-Chernoff converge quickly to the optimal VBR V1 MP3 after seven and four iterations respectively. The others fail to terminate, but the Hoeffding union bounds prune all but two functions while Finite-Sample EMD leaves seven.

Figure \ref{fig:psp-objs} (top right) displays the results of the same experiment applied to the Compression Ratio objective. Because this criterion more evenly dispersed on the interval $[0,1]$ than PEAQ divergence, codecs tend to be pruned quicker as the confidence intervals shrink. Here, VBR V9 MP3 is optimal, and both Asymptotic EMD, Gaussian-Chernoff, and Hoeffding+Union converge, with the approximate bounds again converging nearly instantaneously. The Finite-Sample EMD bounds also prune more effectively, leaving four codecs.

Finally, we apply the algorithm to the combined PEAQ Divergence and Compression Ratio objective in \cref{fig:psp-objs} (bottom left). Because PEAQ divergence and Compression Ratio are contradictory, objective intervals tend to be much harder to separate. The algorithm only terminates for Asymptotic EMD and Gaussian-Chernoff after six and three iterations respectively. Note that VBR V5 MP3 is optimal, which indicates that we're able to successfully choose a ``moderate'' codec that balances the interests of high audio quality and low compression ratio. Finite-Sample EMD and Hoeffding Union each have ten remaining codecs. It struggles to eliminate codecs as efficiently in this case because the conflicts between the two criteria clusters the combined means together. 

\subsection{Analysis}
Across the Progressive Sampling and Global Sampling experiments with different objectives, we find that the Gaussian-Chernoff bounds tend to be the most effective, followed by the Asymptotic EMD, Hoeffding+Union, and Finite-Sample EMD bounds, in that order. It makes logical sense that the two approximate bounds would out-perform their more rigorous counterparts. While the uniform convergence bounds tended to underperform when compared to the union bounds, this is largely due to the small number of compression schemes; because the Hoeffding+Union and Gaussian-Chernoff bounds logarithmically scale with the number of compression schemes $\abs{\mathcal{H}}$, we expect the uniform convergence bounds to perform better in cases with much more than 13 codecs.

Furthermore, changing the objective has a significant impact on the choice of codec. MP3 VBR V1 is consistently the best codec when optimizing for audio quality, MP3 VBR V9 dominates other codecs for compression ratio optimization, and MP3 VBR V5 wins under a combination of the two. This demonstrates our algorithms allow users to choose a codec that fits their explicit preferences that performs near optimally with high probability.

\section{Discussion and Open Questions}
The PSP algorithm provides a clean application of uniform convergence bounds that benefits from locality, by pruning functions that are provably suboptimal. This adaptive nature also leads to improved computational and sample efficiency, as codecs are applied fewer times, and the algorithm terminates before exhausting its data if it can compute the optimal codec to within a given tolerance.  We proved theoretical results that ensure the convergence of our algorithm. We are interested in exploring other options for asymptotic bounds with the goal of obtaining more stable bounds with more provable guarantees. We showed that the framework can be successfully applied to the domain of audio compression. However, no special characteristics of audio data are used here, so we could easily extend this to other function selection tasks.

While we make theoretical arguments about the dominance of bounds based on uniform convergence theory over union bounds, our experiments currently use families of codecs that are too small to exploit those advantages. In future experimental work, we intend to work on incorporating more codecs into the model to enable the EMD bounds the eclipse the union bounds in their ability to prune candidate codecs. Further, we intend to run experiments on samples from different distributions --- such as music files --- to demonstrate that the optimal encoding scheme is highly distribution dependent and thus well-suited to a learning problem.

Because our EMD bounds outperform union bounds when there are a large number of codecs, we think that this algorithm would find a natural application in selecting combinations of encoding schemes. For example, suppose that we wanted to both compress and error-correct a data sequence, and we are unsure about how pairs of encoding schemes will interact with each other. We can use our sampling algorithms to encode each sequence with each pair of schemes to determine which combination of algorithms optimizes an objective based on compression ratio, reconstruction error, and susceptibility to noise. The large number of paired codecs will more starkly highlight the weaknesses of union bounds in this application.

We also hope to incorporate a broader range of criteria; in particular, we want to integrate criteria which synthesizes a range of feedback. For example, in the audio domain the similarity between two sequences is often measured \textit{subjectively} by human listener-provided ratings. Because perceptual abilities varies by listeners, we want to evaluate the listeners as well as the audio files. To do so, we can regard each data point as the pairing of a file and a listener, rather than just a file. We need a different method of complexity that allows samples to drawn with i.i.d. %

One weakness of our techniques is that they only identify a codec that is optimal \emph{on average}.  This can be partially mitigated by constructing objectives using variances or higher moments, which can yield Chernoff-like bounds \citep{philips1995moment}.  However, an alternative strategy would be to uniformly learn the cumulative distribution function of each criterion and each codec, in addition to expectations.  Our methods can easily be extended to this case, and we would then be able to select criteria or induce additional constraints based on provable tail bounds.  This complements our existing framework, and mitigates the chances of selecting a codec that works well on average but occasionally performs extremely poorly.

Our data-driven approach to codec-selection is in line with many trends in machine learning and databases.  We replace fixed codecs with learned codecs, in much the same way that autoML systems replace fixed models with learned pipelines \citep{hutter2011sequential}, and generative adversarial networks \citep{goodfellow2014generative} replace fixed loss functions with learned discriminative loss functions.  Similarly, database systems have recently seen great improvements to query prediction and latency by replacing static indices with learned indices \citep{kraska2018case}.

\subsubsection*{Acknowledgments}

This work was partially supported by NSF award IIS 1813444 and DARPA/US Army grant W911NF-16-1-0553.

{
\small
\setlength{\bibsep}{3pt}
\bibliographystyle{plainnat}%
\bibliography{bib}

\begin{thebibliography}{30}
\providecommand{\natexlab}[1]{#1}
\providecommand{\url}[1]{\texttt{#1}}
\expandafter\ifx\csname urlstyle\endcsname\relax
  \providecommand{\doi}[1]{doi: #1}\else
  \providecommand{\doi}{doi: \begingroup \urlstyle{rm}\Url}\fi

\bibitem[Bartlett and Mendelson(2002)]{bartlett-rademacher}
Peter~L Bartlett and Shahar Mendelson.
\newblock Rademacher and gaussian complexities: Risk bounds and structural
  results.
\newblock \emph{Journal of Machine Learning Research}, 3\penalty0
  (Nov):\penalty0 463--482, 2002.

\bibitem[Bartlett et~al.(2002)Bartlett, Boucheron, and Lugosi]{bartlett-emd}
Peter~L. Bartlett, St\'aphane Boucheron, and G\'abor Lugosi.
\newblock Model selection and error estimation.
\newblock \emph{Machine Learning}, 48\penalty0 (1):\penalty0 85--113, 2002.

\bibitem[Bartlett et~al.(2005)Bartlett, Bousquet, Mendelson,
  et~al.]{bartlett2005local}
Peter~L Bartlett, Olivier Bousquet, Shahar Mendelson, et~al.
\newblock Local rademacher complexities.
\newblock \emph{The Annals of Statistics}, 33\penalty0 (4):\penalty0
  1497--1537, 2005.

\bibitem[Bennett(1962)]{bennett1962probability}
George Bennett.
\newblock Probability inequalities for the sum of independent random variables.
\newblock \emph{Journal of the American Statistical Association}, 57\penalty0
  (297):\penalty0 33--45, 1962.

\bibitem[Boucheron et~al.(2013)Boucheron, Lugosi, and
  Massart]{boucheron2013concentration}
St{\'e}phane Boucheron, G{\'a}bor Lugosi, and Pascal Massart.
\newblock \emph{Concentration inequalities: A nonasymptotic theory of
  independence}.
\newblock Oxford university press, 2013.

\bibitem[Brown et~al.(1971)]{brown1971martingale}
Bruce~M Brown et~al.
\newblock Martingale central limit theorems.
\newblock \emph{The Annals of Mathematical Statistics}, 42\penalty0
  (1):\penalty0 59--66, 1971.

\bibitem[Cantelli(1933)]{cantelli1933sulla}
Francesco~Paolo Cantelli.
\newblock Sulla determinazione empirica delle leggi di probabilita.
\newblock \emph{Giorn. Ist. Ital. Attuari}, 4\penalty0 (421-424), 1933.

\bibitem[Chervonenkis and Vapnik(1971)]{chervonekis}
Alexey Chervonenkis and Vladimir Vapnik.
\newblock On the uniform convergence of relative frequencies of events to their
  probabilities.
\newblock \emph{Theory of Probability and its Applications}, 16:\penalty0
  264--280, 1971.

\bibitem[Dvoretzky et~al.(1956)Dvoretzky, Kiefer, and
  Wolfowitz]{dvoretzky1956asymptotic}
Aryeh Dvoretzky, Jack Kiefer, and Jacob Wolfowitz.
\newblock Asymptotic minimax character of the sample distribution function and
  of the classical multinomial estimator.
\newblock \emph{The Annals of Mathematical Statistics}, pages 642--669, 1956.

\bibitem[Efron and Stein(1981)]{efron1981jackknife}
Bradley Efron and Charles Stein.
\newblock The jackknife estimate of variance.
\newblock \emph{The Annals of Statistics}, pages 586--596, 1981.

\bibitem[Elomaa and K{\"a}{\"a}ri{\"a}inen(2002)]{elomaa2002progressive}
Tapio Elomaa and Matti K{\"a}{\"a}ri{\"a}inen.
\newblock Progressive rademacher sampling.
\newblock In \emph{AAAI/IAAI}, 2002.

\bibitem[Glivenko(1933)]{glivenko1933sulla}
VL~Glivenko.
\newblock Sulla determinazione empirica delle leggi di probabilita.
\newblock \emph{Gion. Ist. Ital. Attauri.}, 4:\penalty0 92--99, 1933.

\bibitem[Goodfellow et~al.(2014)Goodfellow, Pouget-Abadie, Mirza, Xu,
  Warde-Farley, Ozair, Courville, and Bengio]{goodfellow2014generative}
Ian Goodfellow, Jean Pouget-Abadie, Mehdi Mirza, Bing Xu, David Warde-Farley,
  Sherjil Ozair, Aaron Courville, and Yoshua Bengio.
\newblock Generative adversarial nets.
\newblock In \emph{Advances in neural information processing systems}, pages
  2672--2680, 2014.

\bibitem[Gupta and Roughgarden(2017)]{gupta}
Rishi Gupta and Tim Roughgarden.
\newblock A {PAC} approach to application-specific algorithm selection.
\newblock \emph{SIAM Journal on Computing}, 46\penalty0 (3):\penalty0
  992--1017, 2017.

\bibitem[Hegemann et~al.(2017)Hegemann, Leidinger, and Brito]{lame}
Robert Hegemann, Alexander Leidinger, and Rogério Brito.
\newblock {LAME} {MP}3 encoder, 2017.
\newblock URL \url{http://lame.sourceforge.net/}.

\bibitem[Hoeffding(1963)]{hoeffding1963probability}
Wassily Hoeffding.
\newblock Probability inequalities for sums of bounded random variables.
\newblock \emph{Journal of the American statistical association}, 58\penalty0
  (301):\penalty0 13--30, 1963.

\bibitem[Holters and Z\"olzer(2015)]{holters}
Martin Holters and Udo Z\"olzer.
\newblock {GSTPEAQ} -- an open source implementation of the {PEAQ} algorithm.
\newblock \emph{Procedings of the International Conference on Digital Audio
  Effects}, 18, 2015.

\bibitem[Hutter et~al.(2011)Hutter, Hoos, and
  Leyton-Brown]{hutter2011sequential}
Frank Hutter, Holger~H Hoos, and Kevin Leyton-Brown.
\newblock Sequential model-based optimization for general algorithm
  configuration.
\newblock In \emph{International Conference on Learning and Intelligent
  Optimization}, pages 507--523. Springer, 2011.

\bibitem[Jayant et~al.(1993)Jayant, Johnston, and Safranek]{jayant}
Nikil Jayant, James Johnston, and Robert Safranek.
\newblock Signal compression based on models of human perception.
\newblock \emph{Proceedings of the IEEE}, 81\penalty0 (10):\penalty0
  1385--1422, 1993.

\bibitem[Koltchinskii(2001)]{koltchinskii2001rademacher}
Vladimir Koltchinskii.
\newblock Rademacher penalties and structural risk minimization.
\newblock \emph{IEEE Transactions on Information Theory}, 47\penalty0
  (5):\penalty0 1902--1914, 2001.

\bibitem[Kraska et~al.(2018)Kraska, Beutel, Chi, Dean, and
  Polyzotis]{kraska2018case}
Tim Kraska, Alex Beutel, Ed~H Chi, Jeffrey Dean, and Neoklis Polyzotis.
\newblock The case for learned index structures.
\newblock In \emph{Proceedings of the 2018 International Conference on
  Management of Data}, pages 489--504. ACM, 2018.

\bibitem[Massart(1990)]{massart1990tight}
Pascal Massart.
\newblock The tight constant in the dvoretzky-kiefer-wolfowitz inequality.
\newblock \emph{The annals of Probability}, pages 1269--1283, 1990.

\bibitem[McDiarmid(1989)]{mcdiarmid1989method}
Colin McDiarmid.
\newblock On the method of bounded differences.
\newblock \emph{Surveys in combinatorics}, 141\penalty0 (1):\penalty0 148--188,
  1989.

\bibitem[McGuire()]{mcguire}
Hugh McGuire.
\newblock {L}ibri{V}ox: Free public domain audiobooks.
\newblock URL \url{https://librivox.org/}.

\bibitem[Philips and Nelson(1995)]{philips1995moment}
Thomas~K Philips and Randolph Nelson.
\newblock The moment bound is tighter than chernoff's bound for positive tail
  probabilities.
\newblock \emph{The American Statistician}, 49\penalty0 (2):\penalty0 175--178,
  1995.

\bibitem[Riondato and Upfal(2015)]{riondato2015mining}
Matteo Riondato and Eli Upfal.
\newblock Mining frequent itemsets through progressive sampling with rademacher
  averages.
\newblock In \emph{Proceedings of the 21th ACM SIGKDD International Conference
  on Knowledge Discovery and Data Mining}, pages 1005--1014. ACM, 2015.

\bibitem[Riondato and Upfal(2016)]{riondato2016abra}
Matteo Riondato and Eli Upfal.
\newblock Abra: Approximating betweenness centrality in static and dynamic
  graphs with rademacher averages.
\newblock In \emph{Proceedings of the 22nd ACM SIGKDD International Conference
  on Knowledge Discovery and Data Mining}, pages 1145--1154. ACM, 2016.

\bibitem[Riondato and Upfal(2018)]{riondato2018abra}
Matteo Riondato and Eli Upfal.
\newblock Abra: Approximating betweenness centrality in static and dynamic
  graphs with rademacher averages.
\newblock \emph{ACM Transactions on Knowledge Discovery from Data (TKDD)},
  12\penalty0 (5):\penalty0 61, 2018.

\bibitem[Thiede et~al.(2000)Thiede, Treurniet, Bitto, Schmidmer, Sporer,
  Beerends, Colomes, Keyhl, Stoll, Brandenburg, and Feiten]{thiede}
Thilo Thiede, William~C. Treurniet, Roland Bitto, Christian Schmidmer, Thomas
  Sporer, John~G. Beerends, Catherine Colomes, Michael Keyhl, Gerhard Stoll,
  Karlheinz Brandenburg, and Bernhard Feiten.
\newblock {PEAQ}---the {ITU} standard for objective measurement of perceived
  audio quality.
\newblock \emph{Journal of the Audio Engineering Society}, 48\penalty0
  (1):\penalty0 3--29, 2000.

\bibitem[Valiant(1984)]{valiant}
Leslie~G. Valiant.
\newblock A theory of the learnable.
\newblock \emph{Communications of the ACM}, 27\penalty0 (11):\penalty0
  1134--1142, 1984.

\end{thebibliography}
}

\newpage

\appendix
\clearpage{}%

\renewcommand{\thefigure}{A\arabic{figure}}
\setcounter{figure}{0}
\renewcommand{\thetable}{A\arabic{table}}
\setcounter{table}{0}

\section{Proofs}
\label{apdx:sec:proofs}

\subsection{Asymptotic EMD Generalization Bounds}

\thmasymptoticemd*
\begin{proof}
We first show \ref{thm:asymptotic-emd:gaussian}, which follows via a complicated argument regarding the central limit theorem, the Martingale central limit theorem, and the Efron-Stein inequality.  \ref{thm:asymptotic-emd:1tail} and \ref{thm:asymptotic-emd:2tail} then follows from \ref{thm:asymptotic-emd:gaussian} and an asymptotic variant of the EMD symmetrization inequality.

We now show \ref{thm:asymptotic-emd:gaussian}.  We explicitly bound the variance of the EMD; the argument holds as well, mutatis mutandis, for the supremum deviation.

We first use a centering trick on the EMD, which ultimately results in tighter bounds by replacing \emph{raw} variances with \emph{centeralized} variances.  The autocenteredness property of the EMD states that $\EMD{m}{\HC}{\bm{x}} = \EMD{m}{\HC'}{\bm{x}}$, where $\HC' \doteq \setbuilder{h'(x) \doteq h(x) + c(h)}{h \in \HC}$ is equal to $\HC$, except each function is offset by an arbitrary constant $c(h) \in \R$.  Now, take $\HC_{0} \doteq \setbuilder{h_0(x) \doteq h(x) - \Expwrt{x' \distributed \rand{D}}{h(x')}}{h \in \HC}$ to be the \emph{centralized} version of $\HC$.  As $\EMD{m}{\HC}{\bm{x}} = \EMD{m}{\HC_0}{\bm{x}}$ for any $\bm{x} \in \mathcal{X}^m$, we may conclude that their variances are also the same.  The supremum deviation needs no such centering trick, as it is already a centered empirical process.

We now use a standard bound on the variance of suprema of empirical processes, derived from the Efron-Stein inequality \citep{efron1981jackknife}, given as Theorem~11.1 of \citep{boucheron2013concentration}.  %
The \emph{weak variance} $\Sigma^2_{\text{weak}}$ (as defined in \citep{boucheron2013concentration}) of the EMD processes is by definition:

\begin{align*}
\Sigma^2_{\text{weak}} &\doteq \Expwrt{\bm{x} \distributed \rand{D}^m}{\sup_{h' \in \HC'} \sum_{i=1}^m \left(\frac{(-1)^i}{m}h'(\bm{x}_i)\right)^2} & \\
 &= \Expwrt{\bm{x} \distributed \rand{D}^m}{\sup_{h' \in \HC'} \sum_{i=1}^m \frac{1}{m^2}((-1)^i)^2 h'^2(\bm{x}_i)} \\
 &= \frac{1}{m}\Expwrt{\bm{x} \distributed \rand{D}^m}{\sup_{h' \in \HC'} \frac{1}{m}\sum_{i=1}^m h'^2(\bm{x}_i)} & \\
 &= \frac{1}{m}\Expwrt{\bm{x} \distributed \rand{D}^m}{\sup_{h \in \HC} \frac{1}{m}\sum_{i=1}^m (h(\bm{x}_i) - \Expwrt{x' \distributed \rand{D}}{h(x')})^2}
\end{align*}

This quantity is $\frac{1}{m}$ times the expected largest \emph{centered empirical variance estimate} (without Bessel's correction, though asymptotically it matters not) with respect to a sample of size $m$.  By similar logic, it can easily be shown that the \emph{wimpy variance} $\sigma^2_{\text{wimpy}}$ in which the supremum and the expectation are commuted, is $\sigma^2_{\text{wimpy}} \doteq \sup_{h \in \HC} \Varwrt{x'}{h(x')}$.  This quantity is far more convenient than the weak variance, as it is an elementary statistical quantity, well-studied and well-understood, and can reasonably be assumed.  Fortunately, we can argue that for finite $\HC$ and under the bounded variance assumption, they are asymptotically equivalent, as each variance estimate converges to the true variance.

We thus conclude that, asymptotically, %
$\Sigma^2_{\text{weak}} = \sigma^2_{\text{wimpy}} = \frac{1}{m}\sup_{h \in \HC} \Varwrt{x' \distributed \rand{D}}{h(x')}$, and the plugin-estimator for each, which is an unbiased estimator for $\Sigma^2$, but is upward-biased for $\sigma^2$, is $\sup_{h \in \HC}{\Evarp{h(\bm{x})}}$ (with Bessel's correction).  Now, by Theorem~11.1 of \citep{boucheron2013concentration}, we bound the variance of the EMD (and the supremum deviation) as $\sigma^2 = \Sigma^2_{\text{weak}} + \sigma^2_{\text{wimpy}}$.

The final part of the argument is that $\EMD{m}{\HC}{\bm{x}}$ is Gaussian distributed.  This result holds by the Martingale central limit theorem, noting that $m\EMD{m}{\HC}{\bm{x}}$ forms a submartingale, while $m\sqrt{\frac{m}{m+1}}\EMD{m}{\HC}{\bm{x}}$ forms a supermartingale, each with variances that sum to $\infty$.  A similar argument holds for the supremum deviation.

We now show \ref{thm:asymptotic-emd:1tail} and \ref{thm:asymptotic-emd:2tail}.  We require an tighter asymptotic form of the standard EMD symmetrization inequality than the standard finite-sample bound; behold, taking $\bm{x},\bm{x}' \distributed \rand{D}^m$, and letting $\bm{x} \circ \bm{x}'$ be their concatenation:

\begin{align*}
\Expwrt{\bm{x}}{\sup_{h \in \HC} \Expwrt{x' \distributed \rand{D}}{h(x')} - \Eexpp{h(\bm{x})}} &= \Expwrt{\bm{x}}{\sup_{h \in \HC} \Expwrt{\bm{x}'}{\Eexpp{h(\bm{x}')}} - \Eexpp{h(\bm{x})}} \hspace{-0.25cm} & \textsc{Linearity} \\
 &\leq \Expwrt{\bm{x},\bm{x}'}{\sup_{h \in \HC} \Eexpp{h(\bm{x}')} - \Eexpp{h(\bm{x})}} & \textsc{Jensen's Inequality} \\
 &= \Expwrt{\bm{x},\bm{x}'}{\EMD{2m}{\HC}{\bm{x} \circ \bm{x}'}} & \textsc{Definition of } \emd \\
 &\lesssim_m \frac{1}{\sqrt{2}}\Expwrt{\bm{x} \distributed \rand{D}^m}{\EMD{m}{\HC}{\bm{x}}} & \textsc{EMD Asymptotics} \\
\end{align*}
The final step holds, as for finite $\HC$ with bounded variance, by Massart's lemma, $\Expwrt{\bm{x} \distributed \rand{D}^m}{\EMD{2m}{\HC}{\bm{x}}} \leq \RC{m}{\HC}{\bm{x}} \leq \sup_{h \in \HC} \norm{h(\bm{x})}_2 \frac{\sqrt{\log(2\abs{\HC})}}{m} \in \mathcal{O}\left(\frac{1}{\sqrt{m}}\right)$.  %

We are now ready to apply the Chernoff bounds to obtain desiderata \ref{thm:asymptotic-emd:1tail} and \ref{thm:asymptotic-emd:2tail}.  First note that in the bounded-difference EMD tail bounds, we were able to replace the standard double-usage of McDiarmid's inequality with a single usage, improving the $\log(\frac{2}{\delta})$ with $\log(\frac{1}{\delta})$.  Here, a more sophisticated argument may be able to do the same (the EMD and supremum deviation should be correlated, thus their difference should have lower variance than if they were independent), but for simplicity and to require fewer asymptotic assumptions to hold, we instead use a union bound and two applications of the Gaussian Chernoff bound.

The Gaussian Chernoff bound states that $\Probp{\mathcal{N}(\mu, \sigma^2) \geq \mu + \epsilon} \leq e^{-\frac{\epsilon^2}{2\sigma^2}}$.  Applying this to both the supremum deviation and the EMD, bounding the supremum deviation upper-tail and the EMD lower-tail, and combining the results via union bound, we obtain
\[
\Probp{\sup_{h \in \HC} \Expwrt{x'}{h(x')} - \Eexpp{h(\bm{x})} \leq \sqrt{2}\EMD{m}{\HC}{\bm{x}} + (1+\sqrt{2})\epsilon} \leq 2e^{-\frac{m\epsilon^2}{2\sigma^2}}\enspace.
\]
Now, to show \ref{thm:asymptotic-emd:1tail}, we simply take $\delta = \Probp{\sup_{h \in \HC} \Expwrt{x'}{h(x')} - \Eexpp{h(\bm{x})} \leq \sqrt{2}\EMD{m}{\HC}{\bm{x}} + (1+\sqrt{2})\epsilon}$, and compute the minimal $\epsilon$ for which we may guarantee the statement holds with probability at least $1 - \delta$.

\ref{thm:asymptotic-emd:2tail} follows via much the same logic, except we now use a 2-tailed bound on the supremum deviation (a 1-tailed bound suffices for the EMD), and the 3 tails result in $\log(\frac{2}{\delta})$ increasing to $\log(\frac{3}{\delta})$.

\end{proof}

These asymptotic bounds can be converted to finite-sample bounds by using the true variance bound instead of the plugin estimate variance bound, replacing $\sqrt{2}\EMD{m}{\HC}{\bm{x}}$ with $2\EMD{m}{\HC}{\bm{x}}$, and replacing the Gaussian Chernoff bound with Chebyshev's inequality.  Unfortunately, the first summand (the expected variance empirical of the centralized hypothesis class) is quite complicated, and not usually a reasonable quantity to bound \emph{a priori}, though the second (the maximum true variance of the hypothesis class) is well-understood, and assuming a bound is quite reasonable.  Furthermore, the weak polynomial bounds of Chebyshev's inequality, while reasonable for large $\delta$, become prohibitively loose when high-probability tail bounds are desired, though these can be improved by additionally assuming $\sup_{h \in \HC, x,x' \in \mathcal{X}} h(x) - h(x') \leq c$ and applying Bennett's inequality \citep{bennett1962probability}, essentially yielding a hybrid of Theorems~\ref{thm:finite-sample-emd}~\&~\ref{thm:asymptotic-emd}.

\clearpage{}%

\end{document}